\documentclass[preprint,3p,times]{elsarticle}
\usepackage{amsmath}
\usepackage{amsfonts}
\usepackage{amssymb}
\usepackage{tikz}
  \usetikzlibrary{arrows,automata,shapes}
\usepackage{diagbox}
\usepackage{enumitem}

\usepackage{etoolbox}

\makeatletter
\newcommand*{\newbibstartnumber}[1]{%
  \apptocmd{\thebibliography}{%
    \global\c@NAT@ctr #1\relax
    \addtocounter{NAT@ctr}{-1}%
  }{}{}%
}
\makeatother

\newcommand{\fun}[1]{\ensuremath{\textsl{#1}}}

\DeclareMathOperator{\alp}{\rm alph}
\DeclareMathOperator{\A}{\mathcal A}
\DeclareMathOperator{\D}{\mathcal D}

\newcommand{\eps}{\varepsilon}

\newtheorem{theorem}{Theorem}
\newtheorem{lemma}[theorem]{Lemma}

\newdefinition{remark}[theorem]{Remark}
\newdefinition{example}[theorem]{Example}
\newdefinition{problem}[theorem]{Problem}
\newtheorem{claim}{Claim}
\newtheorem{proposition}{Proposition}
\newproof{proof}{Proof}

\newtheorem{fact}{Fact}

\journal{ArXiv.org}

\begin{document}
\begin{frontmatter}

\title{On $k$-piecewise testability (preliminary report)\tnoteref{tr}}

\author{Tom\'{a}\v{s} Masopust\fnref{tm}}
\ead{tomas.masopust@tu-dresden.de}

\author{Micha\"el~Thomazo\fnref{mt}}
\ead{michael.thomazo@tu-dresden.de}

\address{TU Dresden, Germany}

\fntext[tm]{Research supported by the DFG in grant KR~4381/1-1}
\fntext[mt]{Research supported by the Alexander von Humboldt Foundation}
\tnotetext[tr]{This is a full version of the paper accepted for DLT 2015.}

\begin{abstract}
  For a non-negative integer $k$, a language is $k$-piecewise test\-able ($k$-PT) if it is a finite boolean combination of languages of the form $\Sigma^* a_1 \Sigma^* \cdots \Sigma^* a_n \Sigma^*$ for $a_i\in\Sigma$ and $0\le n \le k$. We study the following problem: Given a DFA recognizing a piecewise testable language, decide whether the language is $k$-PT. We provide a complexity bound and a detailed analysis for small $k$'s. The result can be used to find the minimal $k$ for which the language is $k$-PT. We show that the upper bound on $k$ given by the depth of the minimal DFA can be exponentially bigger than the minimal possible $k$, and provide a tight upper bound on the depth of the minimal DFA recognizing a $k$-PT language.
\end{abstract}

\end{frontmatter}

\section{Introduction}
  A regular language is {\em piecewise testable} (PT) if it is a finite boolean combination of languages of the form 
  \[
    \Sigma^* a_1 \Sigma^* a_2 \Sigma^* \cdots \Sigma^* a_n \Sigma^*
  \]
  where $a_i\in \Sigma$ and $n\ge 0$. It is {\em $k$-piecewise testable} ($k$-PT) if $n \le k$. These languages were introduced by Simon in his PhD thesis~\cite{Simon1972}.
  Simon proved that PT languages are exactly those regular languages whose syntactic monoid is $\mathcal{J}$-trivial. He provided various characterizations of PT languages in terms of monoids, automata, etc. 

  In this paper, we study the {\em $k$-piecewise testability} problem, that is, to decide whether a PT language is $k$-PT.
  \begin{itemize}[leftmargin=*]
    \itemsep0pt
    \item[] \textsc{Name: $k$-PiecewiseTestability}
    \item[] \textsc{Input:} an automaton (minimal DFA or NFA) $\A$
    \item[] \textsc{Output: Yes} if and only if $\mathcal L(\A)$ is $k$-piecewise testable
  \end{itemize}
  
  Note that the problem is trivially decidable, since there is only a finite number of $k$-PT languages over the input alphabet of $\A$.
  
  We investigate the complexity of the problem and the relationship between $k$ and the depth of the input automaton. The motivation to study this relationship comes from the result showing that a PT language is $k$-PT for any $k$ bigger than or equal to the depth of its minimal DFA~\cite{KlimaP13}.
  
  Our motivation is twofold. 
  The first motivation is theoretical and comes from the investigation of various fragments of first-order logic over words, namely the Straubing-Th\'erien and dot-depth hierarchies. For instance, the languages of levels 1/2 and 1 of the dot-depth hierarchy are constructed as boolean combinations of variants of languages of the form $\Sigma^* w_1 \Sigma^* \ldots \Sigma^* w_n\Sigma^*$, where $w_i\in\Sigma^*$, cf.~\cite[Table~1]{KufleitnerL12}. The reader can notice a similarity to PT languages. For these fragments, a problem similar to $k$-piecewise testability is also relevant.
  
  The second, practical motivation comes from simplifying the XML Schema specification language. 
  
\paragraph{Simplification of XML Schema}
  XML Schema is currently the only schema language that is widely accepted and supported by industry. However, it is rather machine-readable than human-readable. It increases the expressiveness of DTDs, but this increase goes hand in hand with loss of simplicity. Moreover, its logical core does not seem to be well understood by users~\cite{MartensNNS-vldb12}.
  Therefore, the BonXai schema language has recently been proposed as an attempt to design a human-readable schema language. It combines the simplicity of DTDs with the expressiveness of XML Schema. Its aim is to simplify the development and analysis of XML Schema Definitions (XSDs). 
  The BonXai schema is a set of rules of the form $L_i \to R_i$, where $L_i$ and $R_i$ are regular expressions. An XML document (unranked tree) belongs to the language of the schema if, for every node of the tree, the labels of its children form a word that belongs to $R_i$ and its ancestors form a word that belongs to $L_i$, see~\cite{MartensNNS-vldb12} for more details.
  
  When translating an XSD into an equivalent BonXai schema, the regular expressions $L_i$ are obtained from a finite automaton embedded in the XSD. However, the current techniques of translating automata to regular expressions do not yet generate human-readable results. Therefore, we restrict ourselves to simpler classes of expressions that suffice in practice. Practical and theoretical studies show evidence that expressions of the form $\Sigma^* a_1 \Sigma^* \cdots \Sigma^* a_n$, where $a_i\in \Sigma$, and their variations are suitable for this purpose~\cite{GeladeN11,MartensNSB-tods06}.

  Any state of the DFA embedded in the XSD represents a language and we need to compute an over-approximation $L_i$ for each of them that is disjoint with the others. This reduces to the language separation problem: Given two languages $K$ and $L$ and a family of languages $\mathcal{F}$, is there a language $S$ in $\mathcal{F}$ such that $S$ includes $K$ and is disjoint with $L$? It is independently shown in~\cite{icalp2013} and~\cite{mfcsPlaceRZ13,lvanrooijen} that the separation problem for regular languages represented by NFAs and the family of PT languages is decidable in polynomial time. 
  A simple method (in the meaning of description) to compute a PT separator is described in~\cite{mfcs2014}, where its running time is investigated. Another technique is described in~\cite{mfcsPlaceRZ13}.
  
  Assume that we have computed a PT separator. Since the standard algorithms translating automata to regular expressions do not generate human-readable results and mostly use ``only'' the basic operations (concatenation, Kleene star and union), we face the problem how to generate human-readable expressions of the considered simple forms. Note that the expressions we are interested in contain the operations of intersection and complement (called generalized regular expressions). These operations make them non-elementary more succinct than classical regular expressions~\cite{Dang73,StockmeyerM73}. See also \cite{GeladeN12} for more details. Unfortunately, not much is known about transformations to generalized regular expressions~\cite{EllulKSW05}. 
  
  For a PT language it means to decompose it into a boolean combination of expressions $\Sigma^* a_1 \Sigma^* a_2 \Sigma^* \cdots \Sigma^* a_n \Sigma^*$. If we knew that the language is $k$-PT, this could be derived using a brute-force method and/or the {\em $\sim_k$-canonical DFA}, the DFA whose states are $\sim_k$ classes, cf.~Fact~\ref{mainProperty}. Indeed, the lower the $k$, the lower the complexity. An upper bound on $k$ is given by the depth of the minimal DFA~\cite{KlimaP13}. However, we show later that the minimal $k$ can be exponentially smaller than the depth of the DFA. Note that the number of states of the $\sim_k$-canonical DFA has recently been investigated in~\cite{Karandikar2015} and the literature therein.

\paragraph{Applications of PT Languages}
  Piecewise testable languages are of interest in many disciplines of mathematics and computer science. For instance, in semigroup theory~\cite{Almeida2008486,AlmeidaZ-ita97,PerrinPin}, since they possess interesting algebraic properties, namely, the syntactic monoid of a PT language is $\mathcal{J}$-trivial, where $\mathcal{J}$ is one of the Green relations; in logic over words~\cite{DiekertGK08,PlaceZ_icalp14,PlaceZ14} because of their close relation to first-order logic---piecewise testable languages can be characterized by a (two-variable) fragment of first-order logic over words, namely, they form level~1 of the Straubing-Th\'erien hierarchy as already depicted above; in formal languages and automata theory~\cite{CzerwinskiM14,KlimaP13,mfcsPlaceRZ13}, since their automata are of a special simple form (they are partially ordered and confluent) and PT languages form a strict subclass of the class of star-free languages, that is, languages definable by LTL formulas; in natural language processing, since they can describe some non-local patterns~\cite{FuHT2011,Rogers:2007}; in learning theory, since they are identifiable from positive data in the limit~\cite{GarciaR04,Kontorovich2008}; in XML databases~\cite{icalp2013}, which is our original motivation described in detail above. The list is not comprehensive and many other interesting results concerning PT languages can be found in the literature. It is also worth mentioning that PT languages and several results have recently been generalized from word languages to tree languages~\cite{Bojanczyk:2012}.

  We now give a brief overview on the complexity of the problem to decide whether a regular language is piecewise testable. As mentioned above, decidability was shown by Simon. In 1985, Stern showed that the problem is decidable in polynomial time for DFAs~\cite{Stern85a}. In 1991, Cho and Huynh~\cite{ChoH91} proved NL-completeness of the problem for DFAs. In 2001, Trahtman~\cite{Trahtman2001} improved Stern's result to obtain a quadratic algorithm. Another quadratic algorithm can be found in~\cite{KlimaP13}. The problem is PSPACE-complete if the languages are represented as NFAs~\cite{mfcs2014ex}.

\paragraph{Our Contribution}
  The {\em $k$-piecewise testability problem\/} asks whether, given a finite automaton $\A$, the language $L(\A)$ is $k$-PT. It is easy to see that if a language is $k$-PT, it is also $(k+1)$-PT. Kl\'ima and Pol\'ak~\cite{KlimaP13} have shown that if the depth of a minimal DFA recognizing a PT language is $k$, then the language is $k$-PT. However, the opposite implication does not hold, that is, the depth of the minimal DFA is only an upper bound on $k$. To the best of our knowledge, no efficient algorithm to find the minimal $k$ for which a PT language is $k$-PT nor an algorithm to decide whether a language is $k$-PT has been published so far.\footnote{Very recently, a co-NP upper bound appeared in~\cite{HofmanM15} in terms of separability.}

  We first give a co-NP upper bound to decide whether a minimal DFA recognizes a $k$-PT language for a fixed $k$ (Theorem~\ref{thmconp}), which results in an algorithm to find the minimal $k$ that runs in the time single exponential with respect to the size of the DFA and double exponential with respect to the resulting $k$. We then provide a detailed complexity analysis for small $k$'s. In particular, the problem is trivial for $k=0$, decidable in deterministic logarithmic space for $k=1$ (Theorem~\ref{thm1pt}), and NL-complete for $k=2,3$ (Theorems~\ref{thm2ptNL} and~\ref{thm3ptNL}). 
  As a result, we obtain a PSPACE upper bound to decide whether an NFA recognizes a $k$-PT language for a fixed $k$. Recall that it is PSPACE-complete to decide whether an NFA recognizes a PT language, and it is actually PSPACE-complete to decide whether an NFA recognizes a 0-PT language (Theorem~\ref{thmPSPACE}).
  
  Since the depth of the minimal DFAs plays a role as an upper bound on $k$, we investigate the relationship between the depth of an NFA and $k$-piecewise testability of its language. We show that, for every $k\ge 0$, there exists a $k$-PT language with an NFA of depth $k-1$ and with the minimal DFA of depth $2^k-1$ (Theorem~\ref{thmEXP}). Although it is well known that DFAs can be exponentially larger than NFAs, a by-product of our result is that all the exponential number of states of the DFA form a simple path.
  Finally, we investigate the opposite implication and show that the tight upper bound on the depth of the minimal DFA recognizing a $k$-PT language over an $n$-letter alphabet is $\binom{k+n}{k} - 1$ (Theorem~\ref{tightbound}). A relationship with Stirling cyclic numbers is also discussed.

  For all missing proofs, the reader is referred to the appendix.

\section{Preliminaries and Definitions}
  We assume that the reader is familiar with automata theory~\cite{lawson2003finite}. The cardinality of a set $A$ is denoted by $|A|$ and the power set of $A$ by $2^A$. An alphabet $\Sigma$ is a finite nonempty set. The free monoid generated by $\Sigma$ is denoted by $\Sigma^*$. A word over $\Sigma$ is any element of $\Sigma^*$; the empty word is denoted by $\eps$. For a word $w\in\Sigma^*$, $\alp(w)\subseteq\Sigma$ denotes the set of all letters occurring in $w$, and $|w|_a$ denotes the number of occurrences of letter $a$ in $w$. A language over $\Sigma$ is a subset of $\Sigma^*$. For a language $L$ over $\Sigma$, let $\overline{L}=\Sigma^*\setminus L$ denote the complement of $L$.
  
  A {\em nondeterministic finite automaton\/} (NFA) is a quintuple $\A = (Q,\Sigma,\cdot,I,F)$, where $Q$ is a finite nonempty set of states, $\Sigma$ is an input alphabet, $I\subseteq Q$ is a set of initial states, $F\subseteq Q$ is a set of accepting states, and $\cdot : Q\times\Sigma \to 2^Q$ is the transition function that can be extended to the domain $2^Q \times \Sigma^*$. The language {\em accepted\/} by $\A$ is the set $L(\A) = \{w\in\Sigma^* \mid I \cdot w \cap F \neq \emptyset\}$. We usually omit $\cdot$ and write simply $Iw$ instead of $I\cdot w$. A {\em path\/} $\pi$ from a state $q_0$ to a state $q_n$ under a word $a_1a_2\cdots a_{n}$, for some $n\ge 0$, is a sequence of states and input symbols $q_0 a_1 q_1 a_2 \ldots q_{n-1} a_{n} q_n$ such that $q_{i+1} \in q_i\cdot a_{i+1}$, for all $i=0,1,\ldots,n-1$. The path $\pi$ is {\em accepting\/} if $q_0\in I$ and $q_n\in F$. We use the notation $q_0 \xrightarrow{a_1a_2\cdots a_{n}} q_{n}$ to denote that there exists a path from $q_0$ to $q_n$ under the word $a_1a_2\cdots a_{n}$. A path is {\em simple\/} if all states of the path are pairwise different. The number of states on the longest simple path of $\A$ decreased by one (i.e., the number of transitions on that path) is called the {\em depth\/} of the automaton $\A$, denoted by $\fun{depth}(\A)$. 
  
  The NFA $\A$ is {\em deterministic\/} (DFA) if $|I|=1$ and $|q\cdot a|=1$ for every $q$ in $Q$ and $a$ in $\Sigma$. Then the transition function $\cdot$ is a map from $Q\times\Sigma$ to $Q$ that can be extended to the domain $Q\times\Sigma^*$. Two states of a DFA are {\em distinguishable\/} if there exists a word $w$ that is accepted from one of them and rejected from the other. A DFA is {\em minimal\/} if all its states are reachable and pairwise distinguishable. 

  Let $\A=(Q,\Sigma,\cdot,I,F)$ be an NFA. The reachability relation $\le$ on the set of states is defined by $p\le q$ if there exists a word $w$ in $\Sigma^*$ such that $q\in p\cdot w$. The NFA $\A$ is {\em partially ordered\/} if the reachability relation $\le$ is a partial order. For two states $p$ and $q$ of $\A$, we write $p < q$ if $p\le q$ and $p\ne q$. A state $p$ is {\em maximal\/} if there is no state $q$ such that $p < q$. Partially ordered automata are also called {\em acyclic automata\/}, see, e.g., \cite{KlimaP13}.

  The notion of confluent DFAs was introduced in~\cite{KlimaP13}. Let $\A = (Q, \Sigma, \cdot, i, F)$ be a DFA and $\Gamma \subseteq \Sigma$ be a subalphabet. The DFA $\A$ is $\Gamma$-confluent if, for every state $q$ in $Q$ and every pair of words $u, v$ in $\Gamma^*$, there exists a word $w$ in $\Gamma^*$ such that $(q u) w = (q v) w$. The DFA $\A$ is {\em confluent\/} if it is $\Gamma$-confluent for every subalphabet $\Gamma$. The DFA $\A$ is {\em locally confluent\/} if, for every state $q$ in $Q$ and every pair of letters $a, b$ in $\Sigma$, there exists a word $w$ in $\{a, b\}^*$ such that $(q a) w = (q b) w$.
  
  An NFA $\A=(Q,\Sigma,\cdot,I,F)$ can be turned into a directed graph $G(\A)$ with the set of vertices $Q$, where a pair $(p,q)$ in $Q \times Q$ is an edge in $G(\A)$ if there is a transition from $p$ to $q$ in $\A$. For $\Gamma \subseteq \Sigma$, we define the directed graph $G(\A,\Gamma)$ with the set of vertices $Q$ by considering all those transitions that correspond to letters in $\Gamma$. For a state $p$, let $\Sigma(p)=\{a\in\Sigma \mid p\in p\cdot a\}$ denote the set of all letters under which the NFA $\A$ has a self-loop in the state $p$. Let $\A$ be a partially ordered NFA. If for every state $p$ of $\A$, state $p$ is the unique maximal state of the connected component of $G(\A,\Sigma(p))$ containing $p$, then we say that the NFA satisfies the {\em unique maximal state (UMS) property}.
  
  A regular language is {\em $k$-piecewise testable}, for a non-negative integer $k$, if it is a finite boolean combination of languages of the form $\Sigma^* a_1 \Sigma^* a_2 \Sigma^* \cdots \Sigma^* a_n \Sigma^*$, where $0\le n \le k$ and $a_i\in \Sigma$. A regular language is {\em piecewise testable\/} if it is $k$-piecewise testable for some $k\ge 0$.
  We adopt the notation $L_{a_1 a_2 \cdots a_n} = \Sigma^* a_1 \Sigma^* a_2 \Sigma^* \cdots \Sigma^* a_n \Sigma^*$ from~\cite{KlimaP13}. For two words $v = a_1 a_2 \cdots a_n$ and $w \in L_{v}$, we say that $v$ is a {\em subsequence\/} of $w$ or that $v$ can be {\em embedded\/} into $w$, denoted by $v \preccurlyeq w$. For $k\ge 0$, let $\fun{sub}_k(v) =\{u\in\Sigma^* \mid u\preccurlyeq v,\, |u|\le k\}$. For two words $w_1,w_2$, define $w_1 \sim_k w_2$ if and only if $\fun{sub}_k(w_1)=\fun{sub}_k(w_2)$. If $w_1\sim_k w_2$, we say that $w_1$ and $w_2$ are {\em $k$-equivalent\/}. Note that $\sim_k$ is a congruence with finite index.

  \begin{fact}[\cite{Simon1972}]\label{mainProperty}
    Let $L$ be a regular language, and let $\sim_L$ denote the Myhill congruence~\cite{Myhill}.
    A language $L$ is $k$-PT if and only if $\sim_k\subseteq \sim_L$. 
    Moreover, $L$ is a finite union of $\sim_k$ classes.
  \end{fact}
    
  The theorem says that if $L$ is $k$-PT, then any two $k$-equivalent words either both belong to $L$ or neither does. In terms of minimal DFAs, two $k$-equivalent words lead the automaton to the same state.
  
  \begin{fact}\label{thm:characterization}
    Let $L$ be a language recognized by the minimal DFA $\A$. The following is equivalent.
    \begin{enumerate}
      \itemsep0pt
      \item The language $L$ is PT.
      \item The minimal DFA $\A$ is partially ordered and (locally) confluent~\cite{KlimaP13}.
      \item The minimal DFA $\A$ is partially ordered and satisfies the UMS property~\cite{Trahtman2001}.
    \end{enumerate}
  \end{fact}

\section{Complexity of $k$-Piecewise Testability for DFAs}\label{section3}
  The {\em $k$-piecewise testability problem for DFAs\/} asks whether, given a minimal DFA $\A$, the language $L(\A)$ is $k$-PT. We show that it belongs to co-NP, which can be used to compute the minimal $k$ for which the language is $k$-PT in the time single exponential with respect to the size of the DFA and double exponential with respect to the resulting $k$. For small $k$'s we then provide precise complexity analyses.

  We now prove the following theorem.
  \begin{theorem}\label{thmconp}
    The following problem belongs to co-NP:
    \begin{itemize}[leftmargin=*]
      \itemsep0pt
      \item[] \textsc{Name: $k$-PiecewiseTestability}
      \item[] \textsc{Input:} a minimal DFA $\A$
      \item[] \textsc{Output: Yes} if and only if $\mathcal L(\A)$ is $k$-PT
    \end{itemize}
  \end{theorem}
  
  Let $w_1$ and $w_2$ be two words such that $w_1\preccurlyeq w_2$. Let $\varphi: \{1,2,\ldots,|w_1|\}\to \{1,2,\ldots,|w_2|\}$ be a monotonically increasing mapping induced by one of the possible embeddings of $w_1$ into $w_2$, that is, the letter at the $j$\textsuperscript{th} position in $w_1$ coincides with the letter at the $\varphi(j)$\textsuperscript{th} position in $w_2$. Any such $\varphi$ is called a {\em witness (of the embedding) of $w_1$ in $w_2$}. If we speak about {\em a letter $a$ of $w_2$ that does not belong to the range of $\varphi$}, we mean an occurrence of $a$ in $w_2$ whose position does not belong to the range of $\varphi$.

  \begin{lemma}
  \label{lemma-bound-long}
    Let $\A$ be a minimal DFA recognizing a PT language. If there exist two words $w_1$ and $w_2$ that are $k$-equivalent and lead to two different states from the initial state, such that $w_1$ is a subword of $w_2$, then there exists a $w_2'$ that is $k$-equivalent to $w_1$ leading to the same state as $w_2$ such that $w_2'$ contains at most $\fun{depth}(\A)$ more letters than $w_1$.
  \end{lemma}
  \begin{proof}
    Let us consider $w_1$ and $w_2$ as in the statement of the lemma. Let $\varphi$ be a witness of $w_1$ in $w_2$.  Let $a$ be a letter of $w_2$ that does not belong to the range of $\varphi$. Let us denote $w_2 = w_aaw_a^c$. If $iw_aa = iw_a$, then $iw_aw_a^c = iw_2$. Moreover, since $a \not \in \fun{range}(\varphi)$, $w_1$ is a subword of $w_aw_a^c$. Thus, $\fun{sub}_k(w_1) \subseteq \fun{sub}_k(w_aw_a^c) \subseteq \fun{sub}_k(w_2)$, which proves that $w_1$ and $w_aw_a^c$ are $k$-equivalent. By induction on the number of letters in $w_2$ that do not belong to the range of the given witness of $w_1$ in $w_2$ and that do not trigger a change of state in $\A$, one can show that there exists a word equivalent to $w_1$ and leading to the same state as $w_2$ that does not contain any such letter. Since in a run of an acyclic automaton there are at most $\fun{depth}(\A)$ changes of states, this concludes the proof.
  \qed
  \end{proof}

  \begin{lemma}
  \label{lemma-bound-short}
    Let $\A$ be a minimal DFA recognizing a PT language. If $\mathcal L(\A)$ is not $k$-PT, there exist two words $w_1$ and $w_2$ such that:
    \begin{itemize}
      \itemsep0pt
      \item $w_1$ and $w_2$ are $k$-equivalent;
      \item the length of $w_1$ is at most $k|\Sigma|^k$;
      \item $w_1$ is a subword of $w_2$;
      \item $w_1$ and $w_2$ lead to two different states from the initial state.
    \end{itemize}
  \end{lemma}
  
  \begin{proof}
    If $\mathcal L(\A)$ is not $k$-PT, then there exist $w_1$ and $w_2$ that are $k$-equivalent and lead to two different states from the initial state. Let us show that for $i \in \{1,2\}$, there exists $w_i'$ such that $w_i \sim_k w_i'$ and the length of $w_i'$ is at most $k|\Sigma|^k$. Let $w_i^k$ denote the prefix of $w_i$ of length $k$. Assume that there exists $j$ such that $\fun{sub}_k(w_i^j) = \fun{sub}_k(w_i^{j+1})$. Then the letter at the $({j+1})^\mathrm{th}$ position of $w_i$ can be removed while keeping the same set of subwords of length $k$. Thus there exists $w_i'$ equivalent to $w_i$ such that any two different prefixes of $w_i'$ are not $k$-equivalent. Moreover, since $\fun{sub}_k(w_i^j) \subsetneq \fun{sub}_k(w_i^{j+1})$, such a $w_i'$ contains at most $\sum_{n=1}^{k} |\Sigma|^n \le k|\Sigma|^k$ letters. 
 
    To complete the proof, there are two cases. Either $w_1'$ and $w_2'$ lead to the same state: then, without loss of generality, $w_1'$ and $w_1$ lead to two different states, which proves the claim. Or $w_1'$ and $w_2'$ lead to two different states: then consider $w'$ such that $w' \sim_k w_1'$, and both $w_1'$ and $w_2'$ are subwords of $w'$, which exists by~\cite[Theorem~6.2.6]{SimonS97}. Without loss of generality, $w_1'$ and $w'$ fulfill the required conditions.
  \qed
  \end{proof}
  
  \begin{proof}[of Theorem~\ref{thmconp}]
    One can first check that the automaton $\A$ over $\Sigma$ recognizes a PT language. By Lemma~\ref{lemma-bound-short}, if $\mathcal L(\A)$ is not $k$-PT, there exist two $k$-equivalent words $w_1$ and $w_2$, with the length of $w_1$ being at most $k|\Sigma|^k$, $w_1$ being a subword of $w_2$, and $w_1$ and $w_2$ leading the automaton to two different states. By Lemma~\ref{lemma-bound-long}, one can choose $w_2$ of length at most $\fun{depth}(\A)$ bigger than the length of $w_1$. A polynomial certificate for non $k$-piecewise testability can thus be given by providing such $w_1$ and $w_2$, which are indeed of polynomial length in the size of $\A$ and $\Sigma$.
  \qed
  \end{proof}

  If we search for the minimal $k$ for which the language is $k$-PT, we can first check whether it is 0-PT. If not, we check whether it is 1-PT and so on until we find the required $k$. In this case, the bounds $k|\Sigma|^k$ and $k|\Sigma|^k+\fun{depth}(\A)$ on the length of words $w_1$ and $w_2$ that need to be investigated are exponential with respect to $k$. To investigate all the words up to these lengths then gives an algorithm that is exponential with respect to the size of the minimal DFA and double exponential with respect to the desired $k$.
  
  \begin{proposition}
    Let $\A$ be a minimal DFA that is partially ordered and confluent. To find the minimal $k$ for which the language $L(\A)$ is $k$-PT can be done it time exponential with respect to the size of $\A$ and double exponential with respect to the resulting $k$.
  \end{proposition}

  Theorem~\ref{thmconp} gives an upper bound on the complexity to decide whether a language is $k$-PT for a fixed $k$. We now show that for $k\le 3$, the complexity of the problem is much simpler.

\paragraph{$0$-Piecewise Testability}
 Let $\A$ be a minimal DFA over an alphabet $\Sigma$. The language $L(\A)$ is $0$-PT if and only if it has a single state, that is, it recognizes either $\Sigma^*$ or $\emptyset$. Thus, given a minimal DFA, it is decidable in $O(1)$ whether its language is 0-PT.

\paragraph{$1$-Piecewise Testability}
  \begin{theorem}\label{thm1pt}
    The problem to decide whether a minimal DFA recognizes a 1-PT language is in LOGSPACE.
  \end{theorem}

  The proof of Theorem~\ref{thm1pt} follows immediately from the following lemma.
  \begin{lemma}
    Let $\A=(Q,\Sigma,\cdot,i,F)$ be a minimal DFA. The language $L(\A)$ is 1-PT if and only if both of the following holds:
    \begin{enumerate}
      \item for every $p\in Q$ and $a\in\Sigma$, $pa=q$ implies $qa=q$,
      \item for every $p\in Q$ and $a,b\in\Sigma$, $pab=pba$.
    \end{enumerate}
  \end{lemma}
  \begin{proof}
    We show successively both directions of the equivalence. 

    ($\Rightarrow$) Assume that $L(\A)$ is 1-PT. Since $\mathcal A$ is minimal, $p$ is reachable. Thus, there exists $w$ such that $iw = p$. It holds that $\alp(wa) = \alp(waa)$, thus $wa$ and $waa$ lead to the same state, that is, $qa = q$. Similarly, we notice that $\alp(wab) = \alp(wba)$, and thus $pab = pba$.

    ($\Leftarrow$) We show that for any word $w$, it holds that $iw = ia_1a_2\ldots a_n$, where $\alp(w) = \{a_1,a_2,\ldots,a_n\}$. This then proves that if $w_1 \sim_1 w_2$, then $iw_1 = iw_2$. Thus, since for any letters $a,b \in \Sigma$ and any state $q$, $qab = qba$, we have that $iw = ia_1^{k_1}a_2^{k_2}\ldots a_n^{k_n}$, where $k_i$ is the number of appearances of $a_i$ in $w$. By assumption~$1$ and induction on $k_1 \geq 1$, $ia_1^{k_1} = ia_1$. By induction on $n$, we thus show that $iw = ia_1a_2\ldots a_n$. This shows confluency of $\A$. To show that $\A$ is partially ordered, assume that there exists a cycle $p\xrightarrow{a} q\xrightarrow{w} r\xrightarrow{b} p$, for some states $p\neq q$ and $r$, and a word $w\in\Sigma^*$. By the previous argument, we have that $r = p\cdot aw = p\cdot awa$, that is, $r\cdot a = r$. But then $r\cdot ab = p \neq q = r\cdot ba$, which violates the second assumption.
  \qed\end{proof}

\paragraph{$2$-Piecewise Testability}
  We show that the problem to decide whether a minimal DFA recognizes a 2-PT language is NL-complete. Note that this complexity coincides with the complexity to decide whether the language is PT, that is, whether there exists a $k$ for which the language is $k$-PT.

  \begin{theorem}\label{thm2ptNL}
    The problem to decide whether a minimal DFA recognizes a 2-PT language is NL-complete.
  \end{theorem}

  We first need the following lemma that states that for any two $k$-equivalent words that lead the automaton to two different states, there exist other two equivalent words leading the automaton to two different states, such that one word is a subword of the other and the words differ only by a single letter.

  \begin{lemma}\label{lemma-normal-form}
    Let $\A=(Q,\Sigma,\cdot,i,F)$ be a minimal DFA. For every $k\ge 0$, if $w_1 \sim_k w_2$ and $i w_1 \neq i w_2$, then there exist two words $w$ and $w'$ such that $w \sim_k w'$, $w'$ is obtained from $w$ by adding a single letter at some place, and $i w \neq i w'$.
  \end{lemma}
  \begin{proof}
    Let $w_1$ and $w_2$ be two words such that $w_1 \sim_k w_2$ and $i w_1 \neq i w_2$. Then, by~\cite[Theorem~6.2.6]{SimonS97}, there exists a word $w_3$ such that $w_1$ and $w_2$ are subwords of $w_3$, and $w_1 \sim_k w_2 \sim_k w_3$. Moreover, either $w_1$ and $w_3$, or $w_2$ and $w_3$, do not lead to the same state. Let $v, v' \in \{w_1, w_2, w_3\}$ be such that $v$ is a subword of $v'$ and $i v \neq i v'$. Let $v=u_0, u_1, \ldots, u_n = v'$ be a sequence such that $u_{i+1}$ is obtained from $u_i$ by adding a letter at some place. Such a sequence exists since $v$ is a subword of $v'$. If, for every $i$, $u_i$ and $u_{i+1}$ lead to the same state, then $v$ and $v'$ does as well. Thus, there must exist $i$ such that the words $u_i$ and $u_{i+1}$ lead to two different states and $u_i$ is obtained from $u_{i+1}$ by adding a letter at some place. Setting $w=u_i$ and $w'=u_{i+1}$ completes the proof, since $\fun{sub}_k(v)\subseteq \fun{sub}_k(w)\subseteq \fun{sub}_k(w')\subseteq \fun{sub}_k(v') = \fun{sub}_k(v)$.
  \qed\end{proof}

  \begin{lemma}\label{lemma2ptNL}
    Let $\A=(Q,\Sigma,\cdot,i,F)$ be a minimal partially ordered and confluent DFA. The language $L(\A)$ is 2-PT if and only if for every $a\in\Sigma$ and every states $p$ such that there exists $w$ with $|w|_a \geq 1$, $pua = paua$, for every $u\in\Sigma^*$.
  \end{lemma}

  \begin{proof}
    $(\Rightarrow)$ By contraposition. Assume that there exists $u\in\Sigma^*$ and a state $p$ such that $iw = p$ for some $w\in\Sigma^*$ containing $a$ and such that $p ua \neq p aua$. By the assumption, $w=w_1aw_2$, for some $w_1,w_2\in\Sigma^*$ such that $a\notin\alp(w_1)$, and we want to show that $w_1aw_2ua \sim_2 w_1aw_2aua$. However, for any $c\in\alp(w_1aw_2)$, if $ca\preccurlyeq w_1aw_2aua$, then $ca\preccurlyeq w_1aw_2ua$. Similarly for $d\in\alp(ua)$ and $ad\preccurlyeq w_1aw_2aua$. Since $i\cdot wua \neq i\cdot waua$, the minimality of $\A$ gives that there exists a word $v$ such that $wuav\in L(\A)$ if and only if $wauav\notin L(\A)$. Since $\sim_2$ is a congruence, $wuav \sim_2 wauav$, which violates Fact~\ref{mainProperty}, hence $L(\A)$ is not 2-PT.
    
    $(\Leftarrow)$ 
    Let $w_1$ and $w_2$ be two words such that $w_1 \sim_2 w_2$. We want to show that $i w_1 = i w_2$. By Lemma~\ref{lemma-normal-form}, it is sufficient to show this direction of the theorem for two words $w$ and $w'$ such that $w'$ is obtained from $w$ by adding a single letter at some place. Thus, let $a$ be the letter, and let
    \[
      w = a_1 \ldots a_k a_{k+1} \ldots a_n \text{ and } w' = a_1 \ldots a_k a a_{k+1} \ldots a_n
    \]
    for $0\le k\le n$. Let $w_{i,j}=a_ia_{i+1}\ldots a_j$. We distinguish two cases.
    
    (A) Assume that $a$ does not appear in $w_{1,k}$. Then $a$ must appear in $w_{k+1,n}$. Consider the first occurrence of $a$ in $w_{k+1,n}$. Then $w_{k+1,n} = u_1 a u_2$, where $a$ does not appear in $u_1$. Let $B=\alp(u_1a)$. Then $B\subseteq \alp(u_2)$, because if there is no $a$ in $w_{1,k}u_1$, any subword $ax$, for $x\in B$, that appears in $w'=w_{1,k} a u_1 a u_2$ must also appear in the subword $au_2$ of $w=w_{1,k} u_1 a u_2$. 
    
    Let $u_2 = x_1 b_1 x_2 b_2 x_3 \ldots x_\ell b_\ell x_{\ell+1}$, where $B=\{b_1,b_2,\ldots,b_\ell\}$ and $b_j$ does not appear in $x_1 b_1 x_2 \ldots x_j$, $j=1,2,\ldots,\ell$. Let $v=b_1b_2\ldots b_\ell$. Let $z \in \{ i\cdot w_{1,k} u_1 a$, $i\cdot w_{1,k} a u_1 a\}$. We prove (by induction on $j$) that for every $j=1,2,\ldots,\ell$, there exists a word $y_j$ such that 
    $
      z \cdot (b_1b_2\ldots b_j)^R y_j = z \cdot x_1 b_1 x_2 b_2 x_3 \ldots x_j b_j x_{j+1}.
    $
    Since $b_1$ appears in $u_1$, we use the assumption from the statement of the theorem to obtain $(z\cdot x_1 b_1) \cdot x_2 = (z\cdot b_1 x_1 b_1) \cdot x_2$. Assume that it holds for $j<k$. We prove it for $j+1$. Again, $b_{j+1}$ appears in $u_1$ implies that
    \begin{align*}
      z \cdot x_1 b_1 x_2 b_2 x_3 \ldots x_j b_j x_{j+1} b_{j+1} x_{j+2}
      & = ((z \cdot x_1 b_1 x_2 b_2 x_3 \ldots x_j b_j x_{j+1}) b_{j+1}) x_{j+2} \\
      & = ((z \cdot b_j\ldots b_2b_1 y_j) b_{j+1}) x_{j+2} \\
      & = z \cdot b_{j+1} \underline{b_j\ldots b_2b_1 y_j b_{j+1}} x_{j+2}
    \end{align*}
    where the second equality is by the induction hypothesis and the third is by the assumption from the statement of the theorem applied to the underlined part.
    Thus, in particular, there exists a word $y$ such that $i\cdot w_{1,k} u_1 a v^R y = i\cdot w$ and $i\cdot w_{1,k} a u_1 a v^R y = i\cdot w'$.
    
    Finally, let $z_1 = i\cdot w_{1,k} u_1 a$ and $z_2 = i\cdot w_{1,k} a u_1 a$. We prove that $z_1 \cdot v^R = z_2 \cdot v^R$, which then concludes the proof since it implies that $i\cdot w = i\cdot w'$. To prove this, we make use of the following claim.
    
    \begin{claim}[Commutativity]\label{claimCom}
      For every $a,b\in\Sigma$ and every state $p$ such that $i\cdot w = p$ and $a$ and $b$ appear in $w$, $p\cdot ab = p \cdot ba$.
    \end{claim}
    \begin{proof}
      By the assumption of the theorem, since $a$ appears in $w$, $p\cdot ba = p\cdot aba = q_1$. Similarly, since $b$ appears in $w$, we also have $p\cdot ab = p\cdot bab = q_2$. Then $q_2 \cdot a = (p\cdot ab)a = q_1$ and $q_1\cdot b = (p\cdot ba) b = q_2$. Since the automaton is partially ordered, $q_1=q_2$.
    \hfill$\diamond$\end{proof}
    
    We can now finish the proof by induction on the length of $v^R=b_\ell \ldots b_2 b_1$ by showing that the state $z_i' = z_i\cdot b_\ell \ldots b_2 b_1$ has self-loops under $B$, $i=1,2$. Let $z_i\xrightarrow{b_\ell \ldots b_2 b_1} z_i' = q_{i,\ell+1} b_\ell q_{i,\ell} b_{\ell-1} q_{i,\ell-1} \ldots q_{i,2} b_1 q_{i,1}$ denote the path defined by the word $v^R$ from the state $z_i$, $i=1,2$. 
    
    \begin{claim}
      Both states $z_1'$ and $z_2'$ have self-loops under all letters of the alphabet $B$.
    \end{claim}
    \begin{proof}
      Indeed, $q_{i,j}\cdot b_j = q_{i,j+1}\cdot b_j b_j = q_{i,j+1}\cdot b_j = q_{i,j}$, where the second equality is by the assumption from the statement of the theorem, since $b_j$ appears in $u_1$. Thus, there is a self-loop in $q_{i,j}$ under $b_j$. 
      
      Then, we have $z_i' = q_{i,1} = q_{i,1} b_1 = z_i' b_1$. Now, for every $j=2,\ldots,\ell$, we have $z_i' = q_{i,1} = q_{i,j} \cdot b_{j-1} \ldots b_2 b_1 = q_{i,j} \cdot b_j b_{j-1} \ldots b_2 b_1 = q_{i,j} \cdot b_{j-1} \ldots b_2 b_1 b_j = z_i' b_j$, where the third equality is because there is a self-loop in $q_{i,j}$ under $b_j$, and the fourth is by several applications of commutativity (Claim~\ref{claimCom} above).
    \hfill$\diamond$\end{proof}

    Thus, since no other states are reachable from $z_1'$ and $z_2'$ under $B$, and $z_1'$ and $z_2'$ are reachable from $i\cdot w_{1,k}$ by words over $B$, confluency of the automaton implies that $z_1'=z_2'$, which completes the proof of part (A). 
    
    \medskip
    (B) If $a = a_i$ for some $i \leq k$, we consider two cases. First, assume that for every $c\in\Sigma\cup\{\eps\}$, $ca$ is a subword of $w_{1,k}a$ implies that $ca$ is a subword of $w_{1,k}$. Then $aa$ is a subword of $w_{1,k}$. Let $w_{1,k}=w_3aw_4$, where $a$ does not appear in $w_4$. Let $q = i\cdot w_3 a$, and let $B=\alp(w_4)$. Note that $B\subseteq \alp(w_3)$, since if $xa$ is a subword of $w_{1,k}a$, then it is also in $w_3a$. By the assumption of the theorem, $q = i \cdot w_3 a = i \cdot w_3 a a$, hence we get that there is a self-loop in $q$ under $a$. Now, by the self-loop under $a$ in $q$ and commutativity (Claim~\ref{claimCom} above), 
    $q \cdot w_4 = q \cdot a w_4 = q\cdot w_4 a$. Thus, $i\cdot w_{1,k} = i\cdot w_{1,k} a$.
    
    Second, assume that there exists $c$ in $w_{1,k}$ such that $ca\preccurlyeq w_{1,k}a$ is not a subword of $w_{1,k}$. Then $a$ must appear in $w_{k+1,n}$. Together, there exist $i \le k < j$ such that $a_i=a_j=a$. By the assumption of the theorem, we obtain that $i\cdot w_{1,k} a w_{k+1,j} = i\cdot w_{1,k} w_{k+1,j}$, since $w_{k+1,j}= x a$, for some $x\in\Sigma^*$. This implies that $i\cdot w = i\cdot w'$.

    This completes the proof of part (B) and, hence, the whole proof.
  \qed
  \end{proof}
  
  This result gives a PTIME algorithm to decide whether a minimal DFA recognizes a 2-PT language. 
  However, our aim is to show that the problem is NL-complete. To show that the problem is in NL, we need the following lemma, which gives a characterization of 2-PT languages that can be verified locally in nondeterministic logarithmic space, and provides a quadratic-time algorithm.

  \begin{lemma}\label{lem:equivCond2PT}
    Let $\A=(Q,\Sigma,\cdot,i,F)$ be a DFA. Then the following conditions are equivalent:
      \begin{enumerate}
      \item For every $a\in\Sigma$ and every state $s$ such that $i w = s$ for some $w\in\Sigma^*$ with $|w|_a\ge 1$, $s ua = saua$, for every $u\in\Sigma^*$.
  
      \item For every $a\in\Sigma$ and every state $s$ such that $i w = s$ for some $w\in\Sigma^*$ with $|w|_a\ge 1$, $s ba = s aba$ for every $b\in\Sigma\cup\{\eps\}$.
    \end{enumerate}
  \end{lemma}

  \begin{proof}

    ($1\Rightarrow 2$) 2. is a special case of 1. where $u = b$.  
    
    ($2\Rightarrow 1$) We prove this direction by induction on the length of $u$. Let $a\in \alp(w)$ such that $iw =s$. If $u = \eps$, then we take $b = \eps$. Otherwise, we have $u = u'b$. By induction hypothesis, we have $su'a = sau'a$. Thus $sua = su'ba = (su')ba = (su')aba = (su'a)ba = (sau'a)ba = (sau')ba = saua$.
  \qed 
  \end{proof}

  \begin{proof}[of Theorem~\ref{thm2ptNL}]
    The check of whether a minimal DFA is {\em not} confluent or does {\em not} satisfy condition~2 of Lemma~\ref{lem:equivCond2PT} can be done in NL; the reader is referred to~\cite{ChoH91} for a proof how to check confluency in NL. Since NL=co-NL~\cite{Immerman88,Szelepcsenyi87}, we have an NL algorithm to check 2-piecewise testability of a minimal DFA. NL-hardness follows from the following lemma.
  \qed\end{proof}

  \begin{lemma}
    For every $k\ge 2$, the $k$-PT problem is NL-hard.
  \end{lemma}
  \begin{proof}
    To prove NL-hardness, we reduce an NL-complete problem {\em monotone graph accessibility (2MGAP)\/}~\cite{ChoH91}, which is a special case of the graph reachability problem, to the $k$-piecewise testability problem. An instance of 2MGAP is a graph $(G,s,g)$, where $G=(V,E)$ is a graph with the set of vertices $V=\{1,2,\ldots,n\}$, the source vertex $s=1$ and the target vertex $g=n$, the out-degree of each vertex is bounded by 2 and for all edges $(u,v)$, $v$ is greater than $u$ (the vertices are linearly ordered).
    
    We construct the automaton $\A=(V\cup\{i,f_1,f_2,\ldots,f_{k-1},d\},\Sigma,\cdot,i,\{f_{k-1}\})$ as follows. For every edge $(u,v)$, we construct a transition $u\cdot a_{uv} = v$ over a fresh letter $a_{uv}$. Moreover, we add the transitions $i\cdot a = s$, $g\cdot a = f_1$ and $f_i\cdot a = f_{i+1}$, $i=1,2,\ldots,k-2$, over a fresh letter $a$. The automaton is deterministic, but not necessarily minimal, since some of the states may not be reachable from the initial state, or some states may be equivalent. To ensure minimality of the constructed automaton, we add, for each state $v\in V\setminus\{s\}$, new transitions from $i$ to $v$ under fresh letters, and for each state $v\in V\setminus\{g\}$, new transitions from $v$ to $f_{k-1}$ under fresh letters. All undefined transitions go to the sink state $d$. 
    
    \begin{claim}
      The automaton $\A$ is deterministic and minimal, and $L(\A)$ is finite.
    \end{claim}
    \begin{proof}
      Note that, by construction, all states are reachable from the initial state $i$ and can reach (except the sink state) the unique accepting state $f_{k-1}$. In addition, the automaton is deterministic and minimal, since every transition is labeled by a unique label (except for the transitions $i a = s$ and $g a^{k-1} = f_{k-1}$ labeled with the same letter), which makes the states non-equivalent. Finally, $L(\A)$ is finite because the monotonicity of the graph $(G,s,g)$ implies that the automaton does not contain a cycle nor a self-loop (but the sink state $d$).
      \hfill$\diamond$
    \end{proof}
    
    The following claim is needed to complete the proof.
    \begin{claim}
      Let $w$ be a word over $\Sigma$. If every $a$ from $\Sigma$ appears at most once in $w$, that is, $|w|_a \le 1$, then the language $\{w\}$ is 2-PT.
    \end{claim}
    \begin{proof}
      First, since the language $\{w\}$ is PT, the minimal DFA is partially ordered and confluent. Then the condition of Lemma~\ref{lemma2ptNL} is trivially satisfied, since, after the second occurrence of the same letter, the minimal DFA accepting $\{w\}$ is in the unique maximal non-accepting state.
      \hfill$\diamond$
    \end{proof}
    
    We now show that the language $L(\A)$ is $k$-PT if and only if $g$ is not reachable from $s$. 
    
    By contraposition, we assume that $g$ is reachable from $s$. Let $w$ be a sequence of labels of such a path from $s$ to $g$ in $\A$. Then the word $awa^{k-1}$ belongs to $L(\A)$ and $awa^{k}$ does not. However, $awa^{k-1} \sim_k awa^{k}$, which proves that the language $L(\A)$ is not $k$-PT. 
    
    If $g$ is not reachable from $s$, the language $L(\A)=\{au_1,au_2,\ldots,au_\ell$, $u_{\ell+1},\ldots, u_{\ell+s}\} \cup \{w_1a^{k-1},w_2a^{k-1},\ldots,w_ma^{k-1}\}$, where $u_i$ and $w_i$ are words over $\Sigma\setminus\{a\}$ that do not contain any letter twice. Then the first part is 2-PT by the previous claim, as well as the second part for $k=2$. It remains to show that, for any $k\ge 3$, the second part of $L(\A)$ is $k$-PT. Assume that $w_ja^{k-1} \sim_k w$, for some $1\le j \le m$ and $w\in\Sigma^*$. Then $w= v_1 a v_2 a \ldots a v_k$ for some $v_1,v_2,\ldots,v_k$ such that $|v_1\ldots v_k|_a = 0$. Since $|w_j|_a = 0$ and, for any letter $c$ of $v_2\cdots v_{k-1}$ (resp. $v_k$), the word $aca$ (resp. $a^{k-1}c$) can be embedded into $w_ja^{k-1}$, that is, into $a^{k-1}$, we have that $v_2\cdots v_k=\eps$, i.e., $w=v_1a^{k-1}$. Since $w_ja^{k-1} \sim_k v_1a^{k-1}$, we have that $w_ja = v_1a$ -- hence $w_ja^{k-1}$ and $w$ lead to the same state, concluding the proof.
  \qed\end{proof}

  It was shown in~\cite{BlanchetSadri89} that the syntactic monoids of 1-PT languages are defined by equations $x=x^2$ and $xy=yx$, and those of 2-PT languages by equations $xyzx=xyxzx$ and $(xy)^2=(yx)^2$. These equations can be used to achieve NL algorithms. However, our characterizations improve these results and show that, for 1-PT languages, it is sufficient to verify the equations $x=x^2$ and $xy=yx$ on letters (generators), and that, for 2-PT languages, equation $xyzx=xyxzx$ can be verified on letters (generators) up to the element $y$, which is a general element of the monoid. It decreases the complexity of the problems. Moreover, the partial order and (local) confluency properties can be checked instead of the equation $(xy)^2=(yx)^2$.

\paragraph{$3$-Piecewise Testability}
  The equations $(xy)^3=(yx)^3$, $xzyxvxwy=xzxyxvxwy$ and $ywxvxyzx=ywxvxyxzx$ characterize the variety of 3-PT languages~\cite{BlanchetSadri89}. Non-satisfiability of any of these equations can be check in the DFA in NL by guessing a finite number of states and the right sequences of transitions between them (in parallel, when labeled with the same labels). Thus, we have the following.
  \begin{theorem}\label{thm3ptNL}
    The problem to decide whether a minimal DFA recognizes a 3-PT language is NL-complete.
  \end{theorem}

\paragraph{$k$-Piecewise Testability}
  Even though \cite{BlanchetSadri94} provides a finite sequence of equations to define the $k$-PT languages over a fixed alphabet for any $k\ge 4$, the equations are more involved and it is not clear whether they can be used to obtain the precise complexity. So far, the $k$-piecewise testability problem can be shown to be NL-hard (for $k\ge 2$) and in co-NP, and it is open whether it tends rather to NL or to co-NP.\footnote{See the acknowledgement for the recent development.}

\section{Complexity of $k$-Piecewise Testability for NFAs}
  The {\em $k$-piecewise testability problem for NFAs\/} asks whether, given an NFA $\A$, the language $L(\A)$ is $k$-PT.
  A language is $0$-PT if and only if it is either empty or universal. Since the universality problem for NFAs is PSPACE-complete~\cite{GareyJ79}, the $0$-PT problem for NFAs is PSPACE-complete. Using the same argument as in~\cite{mfcs2014ex} then gives us the following result.
  \begin{proposition}\label{thm0ptnfa}
    For every integer $k\ge 0$, the problem to decide whether an NFA recognizes a $k$-PT language is PSPACE-hard.
  \end{proposition}

  Since $k$ is fixed, we can make use of the idea of Theorem~\ref{thmconp} to decide whether an NFA recognizes a $k$-PT language. The length of the word $w_2$ is now bounded by $2^n$, where $n$ is the number of states of the NFA. Guessing the word $w_2$ on-the-fly then gives that the $k$-piecewise testability problem for NFAs is in PSPACE.
  
  \begin{theorem}\label{thmPSPACE}
    The following problem is PSPACE-complete:
    \begin{itemize}[leftmargin=*]
      \itemsep0pt
      \item[] \textsc{Name: $k$-PiecewiseTestabilityNFA}
      \item[] \textsc{Input:} an NFA $\A$
      \item[] \textsc{Output: Yes} if and only if $\mathcal L(\A)$ is $k$-PT
    \end{itemize}
  \end{theorem}
  \begin{proof}
    Let $\A$ be an NFA over the alphabet $\Sigma$. Let $\A'$ denote the minimal DFA obtained from $\A$ by the standard subset construction and minimization. By Theorem~\ref{thmconp}, and since it is well known that NPSPACE=PSPACE=co-PSPACE, we can guess and store a word $w_1$ of length at most $k|\Sigma|^k$ and to enumerate and store all words of length at most $k$. There are $\sum_{i=1}^{k} |\Sigma|^i$ such words, which is polynomial, since $k$ is a constant. First, we mark all of these words that appear as subwords of $w_1$. Then we guess (letter by letter) a word $w_2$ such that $w_1$ is a subword of $w_2$ (which can be checked by keeping a pointer to $w_1$) and such that the length of $w_2$ is at most $|w_1|+2^{n} = O(2^{n})$, where $n$ is the number of states of the NFA. With each guess of the next letter of $w_2$, we correspondingly move all the pointers to all the stored subwords to keep track of all subwords of $w_2$. We accept if $w_1$ and $w_2$ have the same subwords, $w_1$ is a subword of $w_2$, and $w_1$ and $w_2$ lead the minimal DFA $\A'$ to two different states. Note that because of the space limits the minimal DFA $\A'$ cannot be stored in the memory, but must be simulated on-the-fly while the word $w_2$ is being guessed. The state of $\A'$ defined by the word $w_2$ can then be compared with the state of $\A'$ defined by the word $w_1$, which is either computed at the end or stored from the beginning.
  \qed\end{proof}

  The problem to find the minimal $k$ for which the language recognized by an NFA is $k$-PT is PSPACE-hard, since a language is PT if and only if there exists a minimal $k\ge 0$ for which it is PT.

\section{Piecewise Testability and the Depth of NFAs}\label{secDepth}
  In this section, we generalize a result valid for DFAs to NFAs and investigate the relationship between the depth of an NFA and the minimal $k$ for which its language is $k$-PT. We show that the upper bound on $k$ given by the depth of the minimal DFA can be exponentially far from such a minimal $k$. More specifically, we show that for every $k\ge 0$, there exists a $k$-PT language $L$ recognized by an NFA $\A$ of depth $k-1$ and by the minimal DFA $\D$ of depth $2^k-1$. 

  Recall that a regular language is PT if and only if its minimal DFA satisfies some properties that can be tested in a quadratic time, cf. Fact~\ref{thm:characterization}.
  We now show that this characterization generalizes to NFAs. We say that an NFA $\A$ over an alphabet $\Sigma$ is {\em complete\/} if for every state $q$ of $\A$ and every letter $a$ in $\Sigma$, the set $q\cdot a$ is nonempty, that is, in every state, a transition under every letter is defined.
  
  \begin{theorem}\label{thm10}
    A regular language is PT if and only if there exists a complete NFA that is partially ordered and satisfies the UMS property.
  \end{theorem}
  \begin{proof}
    ($\Rightarrow$) If a regular language is PT, then its minimal DFA
    is partially ordered and satisfies the UMS property by~\cite{Trahtman2001}.
    
    ($\Leftarrow$) To prove the other direction, let $\A=(Q,\Sigma,\cdot,I,F)$ be a complete partially ordered NFA such that it satisfies the UMS property. Let $\D$ be the minimal DFA computed from $\A$ by the standard subset construction and minimization. We represent every state of $\D$ by a set of states of $\A$. 
    
    \begin{claim}\label{claimDminDFA}
      The minimal DFA $\D$ is partially ordered.
    \end{claim}
    \begin{proof}
      Let $X=\{p_1,p_2,\ldots,p_n\}$ with $p_i < p_j$ for $i<j$ be a state of $\D$, and let $w\in\Sigma^*$ be such that $X\cdot w = X$. By induction on $k=1,2,\ldots,n$, we show that $p_i w = p_i$. Assume that for all $i<k$, it holds that $p_i w = p_i$. We prove it for $k$. Since $X=\{p_1,p_2,\ldots,p_n\} = Xw =\cup_{i=1}^{n} p_iw$, $p_k \le p_k w$ and $p_i w = p_i$ for $i<k$, we have that $p_k \in p_k w$. Thus, $\alp(w)\subseteq\Sigma(p_k)$ and the UMS property of $\A$ implies that $p_k w = p_k$. Therefore, for every $a\in\alp(w)$ and $i=1,2,\ldots,n$, $p_i a = p_i$. If, for any state $Y$ of $\D$, $X w_1 = Y$ and $Y w_2 = X$, the previous argument gives that $X=Y$, hence $\D$ is partially ordered.
    \hfill$\diamond$\end{proof}
    
    \begin{claim}
      The minimal DFA $\D$ satisfies the UMS property.
    \end{claim}
    \begin{proof}
      Assume, for the sake of contradiction, that there exist two different states $X$ and $Y$ in the same component of $\D$ that are maximal with respect to the alphabet $\Sigma(X)$. That is, there exist a state $Z$ in $\D$ and two words $u$ and $v$ over $\Sigma(X)$ such that $X = Zu$ and $Y = Zv$. Without loss of generality, we may assume that there exists a state $x$ in $X\setminus Y$. Let $z$ in $Z$ be such that $x = zu$. Since $x$ does not belong to $Y$, $zv\neq x$. Note that $zv$ is defined, since $\A$ is complete. By the proof of the previous claim, $\Sigma(X)\subseteq \Sigma(zv)$ and $\Sigma(X)\subseteq \Sigma(x)$. If $x$ is not reachable from $zv$ by $\Sigma(x)$, we have a contradiction with the UMS property of $\A$. Thus, assume that $zv$ reaches $x$ under $\Sigma(x)$, that is, $zv \le x$. If $x$ does not reach $zv$ under $\Sigma(zv)$, then $zv$ and a maximal state of $x\cdot \Sigma(zv)^*$ are two different maximal states in $\A$, a contradiction. If $x$ reaches $zv$ under $\Sigma(zv)$, then $x \le zv$, which implies, since the NFA is partially ordered, that $zv=x$, which is again a contradiction. 
    \hfill$\diamond$\end{proof}
    Thus, we have shown that the minimal DFA $\D$ is partially ordered and satisfies the UMS property. Fact~\ref{thm:characterization} now completes the proof.
  \qed\end{proof}

  As it is PSPACE-complete to decide whether an NFA defines a PT language, it is PSPACE-complete to decide whether, given an NFA, there is an equivalent complete NFA that is partially ordered and satisfies the UMS property.

\subsection{Exponential Gap between $k$ and the Depth of DFAs}
  It was shown in~\cite{KlimaP13} that the depth of minimal DFAs does not correspond to the minimal $k$ for which the language is $k$-PT. Namely, an example of $(4\ell-1)$-PT languages with the minimal DFA of depth $4\ell^2$, for $\ell > 1$, has been presented. We now show that there is an exponential gap between the minimal $k$ for which the language is $k$-PT and the depth of a minimal DFA.
  
  \begin{theorem}\label{thmEXP}
    For every $n\ge 2$, there exists an $n$-PT language that is not $(n-1)$-PT, it is recognized by an NFA of depth $n-1$, and the minimal DFA recognizing it has depth $2^n-1$.
  \end{theorem}
  \begin{proof}
    For every $k\ge 0$, we define the NFA
    \[
      \A_k=(\{0,1,\ldots,k\},\{a_0,a_1,\ldots,a_k\},\cdot,I_k,\{0\})
    \]
    with $I_k = \{0,1,\ldots,k\}$ and the transition function $\cdot$ consisting of the self-loops under $a_i$ in all states $j > i$ and transitions under $a_i$ from the state $i$ to all states $j < i$. Formally, $i\cdot a_j = i$ if $k\ge i > j \ge 0$ and $i\cdot a_i = \{0,1,\ldots,i-1\}$ if $k\ge i\ge 1$. Automata $\A_2$ and $\A_3$ are shown in Fig.~\ref{fig6}. Note that $\A_k$ is an extension of $\A_{k-1}$, in particular, $L(\A_{k-1})\subseteq L(\A_k)$.
    \begin{figure}[ht]
      \centering
      \begin{tikzpicture}[baseline,->,>=stealth,shorten >=1pt,node distance=1.6cm,
        state/.style={circle,minimum size=1mm,very thin,draw=black,initial text=},
        every node/.style={fill=white,font=\small},
        bigloop/.style={shift={(0,0.01)},text width=1.6cm,align=center}]
        \node[state,initial below,accepting]    (1) {$0$};
        \node[state,initial below]              (4) [left of=1] {$1$};
        \node[state,initial]                    (5) [left of=4] {$2$};
        \path
          (4) edge node {$a_1$} (1)
          (4) edge[loop above] node[bigloop] {$a_0$} (4)
          (5) edge[loop above] node[bigloop] {$a_0,a_1$} (5)
          (5) edge[bend right=60] node {$a_2$} (1)
          (5) edge node {$a_2$} (4);
      \end{tikzpicture}
      \qquad
      \begin{tikzpicture}[baseline,->,>=stealth,shorten >=1pt,node distance=1.6cm,
        state/.style={circle,minimum size=1mm,very thin,draw=black,initial text=},
        every node/.style={fill=white,font=\small},
        bigloop/.style={shift={(0,0.01)},text width=1.6cm,align=center}]
        \node[state,initial below,accepting]    (0) {$0$};
        \node[state,initial below]              (1) [left of=0] {$1$};
        \node[state,initial below]              (2) [left of=1] {$2$};
        \node[state,initial]                    (3) [left of=2] {$3$};
        \path
          (3) edge node {$a_3$} (2)
          (2) edge node {$a_2$} (1)
          (1) edge node {$a_1$} (0)
          (3) edge[bend right=55] node {$a_3$} (1)
          (3) edge[bend right=65] node {$a_3$} (0)
          (2) edge[bend right=55] node {$a_2$} (0)
          (3) edge[loop above] node[bigloop] {$a_0, a_1, a_2$} (3)
          (2) edge[loop above] node[bigloop] {$a_0, a_1$} (2)
          (1) edge[loop above] node[bigloop] {$a_0$} (1);
      \end{tikzpicture}
      \caption{Automata $\A_2$ and $\A_3$.}
      \label{fig6}
    \end{figure}
    
    We define the word $w_k$ inductively by $w_0=a_0$ and $w_{\ell}=w_{\ell-1}a_{\ell}w_{\ell-1}$, for $0 < \ell \le k$. Note that $|w_\ell|=2^{\ell+1}-1$. In~\cite{mfcs2014ex}, we have shown that every prefix of $w_k$ of odd length ends with $a_0$ and, thus, does not belong to $L(\A_k)$, while every prefix of even length belongs to $L(\A_k)$. For convenience, we briefly recall the proof here. The empty word belongs to $L(\A_0)\subseteq L(\A_k)$. Let $v$ be a prefix of $w_k$ of even length. If $|v| < 2^k-1$, then $v$ is a prefix of $w_{k-1}$ and, by the induction hypothesis, $v\in L(\A_{k-1})\subseteq L(\A_k)$. If $|v| > 2^k-1$, then $v = w_{k-1} a_k v'$. The definition of $\A_k$ and the induction hypothesis then yield that there is a path $k\xrightarrow{w_{k-1}} k\xrightarrow{a_k}\,(k-1) \xrightarrow{v'} 0$. Thus, $v$ belongs to $L(\A_k)$.

    We now discuss the depth of the minimal DFA recognizing the language $L(\A_k)$.
    \begin{claim}
      For every $k\ge 0$, the depth of the minimal DFA recognizing the language $L(\A_k)$ is $2^{k+1}-1$.
    \end{claim}
    \begin{proof}
      We prove the claim by induction on $k$. For $k=0$, the minimal DFA $\fun{det}(\A_0)=(\{\{0\},\emptyset\},\{a_0\},\cdot,\{0\},\{0\})$ obtained from $\A_0$ by the standard subset construction and minimization has two states, accepts the single word $\eps$, and $a_0$ goes from the initial state $I_0 =\{0\}$ to the sink state $\emptyset$. Thus, it has depth $1$ as required. Consider the word $w_k=w_{k-1}a_kw_{k-1}$ for $k>0$. By the induction hypothesis, there exists a simple path of length $2^{k}-1$ in $\fun{det}(\A_{k-1})$ defined by the word $w_{k-1}$ starting from the initial state $I_k = \{0,1,\ldots,k-1\}$ and ending in the state $\emptyset$. Let $Q_0,Q_1,\ldots, Q_{2^k-1}$ denote the states of that simple path in the order they appear on the path, that is, $Q_0= I_k$, $Q_{2^k-1} = \emptyset$, and $Q_i\subseteq Q_0$ for $i=1,2,\ldots,2^k-1$. Note that the states are pairwise non-equivalent by the induction hypothesis. Let $w_{k-1,i}$ denote the $i$-th letter of the word $w_{k-1}$. Then the path
      \begin{align*}
        \underbrace{(Q_0\cup\{k\})\xrightarrow{w_{k-1,1}} (Q_1\cup\{k\})
        \xrightarrow{w_{k-1,2}} (Q_2\cup\{k\})
        \xrightarrow{\ \ldots\ } (Q_{2^k-1}\cup\{k\})}_{w_{k-1}}
        \xrightarrow{a_k}
        \underbrace{Q_0\xrightarrow{w_{k-1,1}} Q_1
        \xrightarrow{w_{k-1,2}} Q_2
        \xrightarrow{\ \ldots\ } Q_{2^k-1}}_{w_{k-1}}
      \end{align*}
      consists of $2^{k+1}$ different states. We show that these states are pairwise non-equivalent. Since the letter $a_k$ is accepted from every state $Q_j\cup\{k\}$, but from no state $Q_i$, for $0\le i,j\le 2^k-1$, the state $Q_j\cup\{k\}$ is distinguishable from the state $Q_i$. Moreover, $Q\cup\{k\}$ and $Q'\cup\{k\}$ are distinguished by the same word as the states $Q$ and $Q'$, that are distinguishable by the induction hypothesis. Thus, we have a simple path of length $2^{k+1}-1$ as required.
      \unskip\hfill$\diamond$
    \end{proof}

    We now show that $\A_k$ defines a $(k+1)$-PT language that is not $k$-PT.
    \begin{claim}
      For every $k\ge 0$, the language $L(\A_k)$ is $(k+1)$-PT.
    \end{claim}
    \begin{proof}
      By induction on $k$. For $k=0$, the language $L(\A_0)=\{\eps\}=\cap_{a\in\Sigma} \overline{L_a}$ is indeed 1-PT. Consider the automaton $\A_k$ and let $u$ and $v$ be two words such that $u \sim_{k+1} v$. Assume that $u\in L(\A_k)$. We show that $v\in L(\A_k)$ as well. 
      If $u$ does not contain the letter $a_k$, then $u\in L(\A_{k-1})$ and, since $u \sim_{k+1} v$ implies that $u \sim_k v$, the induction hypothesis gives that $v\in L(\A_{k-1})\subseteq L(\A_k)$. 
      If $u$ contains the letter $a_k$, the definition of $\A_k$ gives that $u$ is of the form $u = u_1 a_k u_2$, where $u_1u_2$ does not contain the letter $a_k$. Since $u\sim_{k+1} v$, the word $v$ is also of a form $v = v_1 a_k v_2$, where $v_1v_2$ does not contain the letter $a_k$. However, $u_2 \sim_k v_2$, since $w\in \fun{sub}_k(u_2)$ if and only if $a_kw \in \fun{sub}_{k+1}(u_1a_ku_2)=\fun{sub}_{k+1}(v_1a_kv_2)$, which is if and only if $w\in \fun{sub}_k(v_2)$. Since, by the induction hypothesis, $u_2\in L(\A_{k-1})$ implies that $v_2\in L(\A_{k-1})$, we obtain that $v\in L(\A_k)$.
      \hfill$\diamond$
    \end{proof}
    
    \begin{claim}
      For every $k\ge 0$, the language $L(\A_k)$ is not $k$-PT. 
    \end{claim}
    \begin{proof}
      Let $w_k = w_{k-1} a_k w_{k-1}$ be the word defined above. Let $w_k'$ denote the prefix of $w_k$ without the last letter (which is $a_0$), that is, $w_k = w_k' a_0$. We now show, by induction on $k$, that $w_k \sim_k w_k'$. This then implies that the language $L(\A_k)$ is not $k$-PT, because $w_k'$ belongs to $L(\A_k)$ while $w_k$ does not belong to $L(\A_k)$. Indeed, for $k=0$, we have $w_0 = a_0 \sim_0 \eps = w_0'$. Thus, assume that $w_k \sim_k w_k'$ for some $k\ge 0$, and consider a word $w\in \fun{sub}_{k+1}(w_k a_{k+1} w_k)$. Then the word $w$ can be decomposed to $w = w'w''$, where $w'$ is the maximal prefix of $w$ that can be embedded into the word $w_ka_{k+1}$. Note that $w''$ is a suffix of $w$ that can be embedded into $w_k$. Since $|w'|>0$, we have that $|w''|\le k$. By the induction hypothesis, $w''\in \fun{sub}_k(w_k)=\fun{sub}_k(w_k')$. Thus, $w=w'w''\in \fun{sub}_{k+1}(w_ka_{k+1}w_k')$, which proves that $w_{k+1} \sim_{k+1} w_{k+1}'$.
      \hfill$\diamond$
    \end{proof}

    To finish the proof of Theorem~\ref{thmEXP}, note that every NFA $\A_k$ has depth $k$, accepts a $(k+1)$-PT language that is not $k$-PT and its minimal DFA has depth $2^{k+1}-1$. This completes the proof.
  \qed\end{proof}

  Although it is well known that DFAs can be exponentially larger than NFAs, an interesting by-product of this result is that there are NFAs such that all the exponential number of states of their minimal DFAs form a simple path.

  It could seem that NFAs are more convenient to provide upper bounds on the $k$. However, the following simple example demonstrates that even for 1-PT languages, the depth of an NFA depends on the size of the input alphabet. Specifically, for any alphabet $\Sigma$, the language $L=\bigcap_{a\in\Sigma} L_a$ of all words containing all letters of $\Sigma$ is a $1$-PT language such that any NFA recognizing it requires at least $2^{|\Sigma|}$ states and has depth $|\Sigma|$. A deeper investigation in this direction is provided in the next section.

  \begin{example}\label{exSimple}
    Let $L=\bigcap_{a\in\Sigma} L_a$ be a language of all words that contain all letters of the alphabet. Then $2^{|\Sigma|}$ states are sufficient for an NFA to recognize $L$. Indeed, the automaton $\A=(2^\Sigma,\Sigma,\cdot,\{\emptyset\},\{\Sigma\})$ with the transition function defined by $X\cdot a = X \cup \{a\}$, for $X \subseteq \Sigma$ and $a \in \Sigma$, recognizes $L$. The depth of $\A$ is $|\Sigma|$, since every non-self-loop transition goes to a strict superset of the current state.
    
    To prove that every NFA requires at least $2^{|\Sigma|}$ states, we use a fooling set lower-bound technique~\cite{bi92}. A set of pairs of words $\{(x_1,y_1),(x_2,y_2),\ldots,(x_n,y_n)\}$ is a fooling set for $L$ if, for all $i$, the words $x_iy_i$ belong to $L$ and, for $i\neq j$, at least one of the words $x_iy_j$ and $x_jy_i$ does not belong to $L$. To construct such a fooling set, for any $X\subseteq \Sigma$, we fix a word $w_X$ such that $\alp(w_X)=X$. Let $S=\{(w_X,w_{\Sigma\setminus X}) \mid X\subseteq \Sigma\}$. Then $\alp(w_Xw_{\Sigma\setminus X}) =\Sigma$ and $w_Xw_{\Sigma\setminus X}$ belongs to $L$. On the other hand, for $X\neq Y$, either $X\cup (\Sigma\setminus Y)$ or $Y\cup (\Sigma\setminus X)$ is different from $\Sigma$, which implies that $S$ is a fooling set of size $2^{|\Sigma|}$. The main result of~\cite{bi92} now implies the claim.
    It remains to prove that the depth is at least $|\Sigma|$. However, the shortest words of $L$ are of length $|\Sigma|$, which completes the proof.

    Note that if we consider union instead of intersection, the resulting minimal DFA has only $2$ states and depth $1$. 
  \end{example}

\section{Tight Bounds on the Depth of Minimal DFAs}
  If a PT language is recognized by a minimal DFA of depth $\ell$, then it is $\ell$-PT. However, the opposite implication does not hold and the analysis of Section~\ref{secDepth} shows that the language can be $(\ell-i)$-PT for exponentially large $i$'s. Therefore, we study the opposite implication of the relationship between $k$-piecewise testability and the depth of the minimal DFA in this section. Specifically, given a $k$-PT language over an $n$-letter alphabet, we show that the depth of the minimal DFA recognizing it is at most $\binom{k+n}{k} - 1$.

  To this end, we first investigate the following problem. 
  \begin{problem}\label{problem1}
    Let $\Sigma$ be an alphabet of cardinality $n\ge 1$ and let $k\ge 1$. What is the length of a longest word, $w$, such that
      $\fun{sub}_k(w) = \Sigma^{\le k} = \{ v \in\Sigma^* \mid |v| \le k\}$ and,
      for any two distinct prefixes $w_1$ and $w_2$ of $w$, $\fun{sub}_k(w_1)\neq \fun{sub}_k(w_2)$?
  \end{problem}
  
  The answer to this question is formulated in the following proposition proved below by two lemmas.
  \begin{proposition}\label{propSolution}
    Let $\Sigma$ be an alphabet of cardinality $n$. The length of a longest word, $w$, satisfying the requirements of Problem~\ref{problem1} is given by the recursive formula 
    $
      |w| = P_{k,n} = P_{k-1,n} + P_{k,n-1} + 1,
    $
    where $P_{1,m} = m = P_{m,1}$, for $m\ge 1$.
  \end{proposition}
  
  Equivalently stated, Problem~\ref{problem1} asks what is the depth of the {\em $\sim_k$-canonical DFA}, whose states correspond to $\sim_k$ classes, that is, of a DFA $\A = (Q,\Sigma,\cdot,[\eps],F)$, where $Q=\left\{ [w] \mid w\in\Sigma^{\le k}\right\}$, $[w]=\{ w' \mid w'\sim_k w\}$, and the transition function $\cdot$ is defined so that, for a state $[w]$ and a letter $a$, $[w] \cdot a = [wa]$. The set of accepting states $F$ is not important here, but will be used later.

  We show below that the solution to this problem is given by the following recursive formula:
  \[
    |w| = P_{k,n} = P_{k-1,n} + P_{k,n-1} + 1\,,
  \]
  where $P_{1,m} = m = P_{m,1}$, for any $m\ge 1$.
  
  The following lemma shows that $w$ is not longer than $P_{k,n}$.
  \begin{lemma}\label{lowerbound}
    Let $k$ and $n$ be given, and let $w'$ be any word over an $n$-letter alphabet satisfying the requirements of Problem~\ref{problem1}. Then $|w'| \le P_{k,n}$.
  \end{lemma}
  \begin{proof}
    Let $w'$ be a word over $\Sigma=\{a_1,a_2,\ldots,a_n\}$ with the order $a_i<a_j$ if $i<j$ induced by the occurrence of $a$ in $w'$. For instance, $abadca$ induces the order $a<b<d<c$. Let $z$ denote the first occurrence of $a_n$ in $w'$. Then $w' = w_1 z w_2$, where $w_1$ is a word over $\{a_1,a_2,\ldots,a_{n-1}\}$ satisfying the second requirement of Problem~\ref{problem1}, hence $|w_1|\le P_{k,n-1}$. On the other hand, since $\alp(w_1z)=\Sigma$, any prefix of $w_2$ extends the set of subwords with a subword of length at least 2. Thus, $w_2$ cannot be longer than the longest word over $\Sigma$ containing all subwords up to length $k-1$, that is, $|w_2| \le P_{k-1,n}$. This completes the proof.
  \qed\end{proof}
  
  We now show that there exists a word of length $P_{k,n}$.
  \begin{lemma}\label{upperbound}
    For any positive integers $k$ and $n$, there exists a word $w$ of length $P_{k,n}$ satisfying the requirements of Problem~\ref{problem1}.
  \end{lemma}
  \begin{proof}
    Let $\Sigma_n$ denote the alphabet $\{a_1,a_2,\ldots, a_n\}$ with the order $a_i < a_j$ if $i<j$. For $n=1$ and $k\ge 1$, the word $W_{k,1} = a^k$ is of length $P_{k,1}$ and satisfies the requirements, as well as the word $W_{1,n} = a_1a_2\ldots a_n$ of length $P_{1,n}$ for $k=1$ and $n\ge 1$. Assume that we have constructed the words $W_{i,j}$ of length $P_{i,j}$ for all $i < k$ and $j < n$, $W_{i,n}$ of length $P_{i,n}$ for all $i < k$, and $W_{k,j}$ of length $P_{k,j}$ for all $j < n$. We construct the word $W_{k,n}$ of length $P_{k,n}$ over $\Sigma_{n}$ as follows:
    \[
      W_{k,n} = W_{k,n-1}\, a_{n}\, W_{k-1,n}\,.
    \]
    
    It remains to show that $W_{k,n}$ satisfies the requirements of Problem~\ref{problem1}. However, the set of subwords of $W_{k-1,n}$ is $\Sigma_n^{\le k-1}$. Since $\alp(W_{k,n-1}a_{n})=\Sigma$, we obtain that the set of subwords of $W_{k,n}$ is $\Sigma_n^{\le k}$.
    
    Let $w_1$ and $w_2$ be two different prefixes of $W_{k,n}$. Without loss of generality, we may assume that $w_1$ is a prefix of $w_2$. If they are both prefixes of $W_{k,n-1}$, the second requirement of Problem~\ref{problem1} follows by induction. If $w_1$ is a prefix of $W_{k,n-1}$ and $w_2$ contains $a_n$, then the second requirement of Problem~\ref{problem1} is satisfied, because $w_1$ does not contain $a_n$. Thus, assume that both $w_1$ and $w_2$ contain $a_n$, that is, they both contain $W_{k,n-1} a_n$ as a prefix. Let $w_1=W_{k,n-1}a_n w_1'$ and $w_2=W_{k,n-1}a_n w_1'w_2'$. Since, by induction, $\fun{sub}_{k-1}(w_1')\subsetneq\fun{sub}_{k-1}(w_1'w_2')$, there exists $v \in \fun{sub}_{k-1}(w_1'w_2') \setminus \fun{sub}_{k-1}(w_1')$. Then $a_n v$ belongs to $\fun{sub}_k(w_2)$, but not to $\fun{sub}_k(w_1)$, which completes the proof.
  \qed\end{proof}

  It follows by induction that for any positive integers $k$ and $n$
  \begin{align}\label{def-pkn}
      P_{k,n} = \binom{k+n}{k} - 1\,.
  \end{align}

  We now use this result to show that the depth of the minimal DFA recognizing a $k$-PT language over an $n$-letter alphabet is $P_{k,n}$ in the worst case.
  \begin{theorem}\label{tightbound}
    For any natural numbers $k$ and $n$, the depth of the minimal DFA recognizing a $k$-PT language over an $n$-letter alphabet is at most $P_{k,n}$. Moreover, the bound is tight for any $k$ and $n$.
  \end{theorem}
  \begin{proof}
    Let $L_{k,n}$ be a $k$-PT language over an $n$-letter alphabet. Since $L_{k,n}$ is a finite union of $\sim_k$ classes~\cite{Simon1972}, there exists $F$ such that the $\sim_k$-canonical DFA $\A = (Q,\Sigma,\cdot,[\eps],F)$ recognizes $L_{k,n}$. The depth of $\A$ is $P_{k,n}$. Let $\fun{min}(\A)$ denote the minimal DFA obtained from $\A$ by a standard minimization procedure. Since the minimization does not increase the depth, the depth of $\fun{min}(\A)$ is at most~$P_{k,n}$.
    
    To show that the bound is tight, let $w$ denote a fixed word of length $P_{k,n}$, which exists by Lemma~\ref{upperbound}. Consider the $\sim_k$-canonical DFA $\A' = (Q,\Sigma,\cdot,[\eps],F)$, where $F=\{ [w'] \mid w' \text{ is a prefix of } w \text{ of even length}\}$. Then $w$ defines a path $\pi_w=[\eps]\xrightarrow{w_1} [w_1] \xrightarrow{w_2} [w_2] \ldots \xrightarrow{w} [w]$ in $\A'$ of length $P_{k,n}$, where $w_i$ denotes the prefix of $w$ of length $i$ and accepting and non-accepting states alternate. Again, let $\fun{min}(\A')$ denote the minimal DFA obtained from $\A'$. If there were two equivalent states in $\pi_w$, then they must be of the same acceptance status. However, between any two states with the same acceptance status, there exists a state with the opposite acceptance status. Therefore, joining the two states creates a cycle in $\fun{min}(\A')$, which is a contradiction with Fact~\ref{thm:characterization}, since the DFA $\A'$ recognizes a PT language.
  \qed\end{proof}

  A few of these numbers are listed in Table~\ref{table1}.
  \begin{table}[b]
  \begin{center}
    \begin{tabular}{|c||r|r|r|r|r|r|}
      \hline
      \diagbox{k}{n}
              & n=1 & n=2 & n=3 & n=4 & n=5 & n=6 \\\hline\hline
          k=1 & 1   & 2   & 3   & 4   & 5   & 6\\\hline
          k=2 & 2   & 5   & 9   & 14  & 20  & 27\\\hline
          k=3 & 3   & 9   & 19  & 34  & 55  & 83\\\hline
          k=4 & 4   & 14  & 34  & 69  & 125 & 209\\\hline
          k=5 & 5   & 20  & 55  & 125 & 251 & 461\\\hline
          k=6 & 6   & 27  & 83  & 209 & 461 & 923\\\hline
    \end{tabular}
  \end{center}
  \caption{The table of a few first numbers $P_{k,n}$}
  \label{table1}
  \end{table}
  We now present several consequences of these results.
  
  \begin{enumerate}
    \item Note that it follows from the formula that $P_{k,n} = P_{n,k}$. This gives and interesting observation that increasing the length of the considered subwords has exactly the same effect as increasing the size of the alphabet.
    
    \item Equivalently stated, Problem~\ref{problem1} asks what is the depth of the $\sim_k$-canonical DFA, whose states are $\sim_k$ classes. The number of equivalence classes of $\sim_k$, i.e., the number of states, has recently been investigated in~\cite{Karandikar2015}.
  
    \item It provides a precise bound on the length of $w_1$ of Theorem~\ref{thmconp}. However, it does not improve the statement of the theorem.
  \end{enumerate}
  
  To provide a relationship of $P_{k,n}$ with Stirling cyclic numbers, the following can be shown.
  \begin{proposition}
    For positive integers $k$ and $n$, $P_{k,n} = \frac{1}{k!}\sum_{i=1}^k {k+1 \brack i+1}n^i$, where ${k \brack n}$ denotes the Stirling cyclic numbers.
  \end{proposition}
  \begin{proof}
    To prove this, we first recall the following well-known properties of Stirling cyclic numbers.
    \begin{align}\label{eq-stirling-sum}
      {k+1 \brack 1} = k! & & \text{ and } &&  \sum_{i=0}^{k} {k \brack i} x^i = x(x+1)\cdots (x+k-1) = \frac{(x+k-1)!}{(x-1)!}
    \end{align}
    Now we prove the claim.
    \begin{align*} 
      \frac{1}{k!}\sum\limits_{i=1}^k {k+1 \brack i+1}n^i 
        &= \frac{1}{nk!}\sum\limits_{i=1}^k {k+1 \brack i+1}n^{i+1} \\
        & \textrm{(multiplication by $n/n$)} \\
        &= \frac{1}{nk!}\sum\limits_{i=2}^{k+1} {k+1 \brack i}n^{i} \\
        & \textrm{(changing indexes)} \\
        &= \frac{1}{nk!}\left(\sum\limits_{i=0}^{k+1} {k+1 \brack i}n^{i}  - {k+1 \brack 1}n\right)\\
        & \textrm{(adding the cases $i=0,1$ into the sum)} \\
        &= \frac{1}{nk!}\left(\frac{(k+n)!}{(n-1)!} - k!n\right)\\
        & \textrm{(by Equation~\ref{eq-stirling-sum})} \\
        &= \frac{(k+n)!}{n!k!} - 1 \\
        &= P_{k,n}\\  
        & \textrm{(by Equation~\ref{def-pkn})}
    \end{align*}
    This completes the proof.
  \qed\end{proof}

  Finally, note that one could also see a noticeable relation between the columns (resp. rows) of Table~\ref{table1} and the generalized Catalan numbers of~\cite{frey2001}. We leave the details of this correspondence for a future investigation.

\paragraph*{Acknowledgements.}
  We thank an anonymous reviewer for informing us about the unpublished manuscript~\cite{KKP} and its authors for providing it. It shows that the $k$-PT problem is co-NP-complete for $k\ge 4$. It also provides a smaller bound on the length of the witnesses, which results in a single exponential algorithm to find the minimal $k$.
  
  The authors are grateful to Sebastian Rudolph for a fruitful discussion.

\end{document}